\documentclass[12pt,cls,onecolumn]{IEEEtran}
\usepackage{subfigure,cite,graphicx,psfig,amsmath,amssymb,mathrsfs,epsf}
\usepackage{cite}
\usepackage{graphicx}
\usepackage{psfrag}
\usepackage{subfigure}
\usepackage{url}
\usepackage{stfloats}
\usepackage{amsmath}
\usepackage{amssymb}
\usepackage{array}

\begin{document}

\title{Opportunistic Interference Mitigation Achieves Optimal Degrees-of-Freedom in Wireless Multi-cell Uplink Networks}
\author{\large Bang Chul Jung, \emph{Member}, \emph{IEEE}, Dohyung Park, and Won-Yong Shin, \emph{Member}, \emph{IEEE} \\
\thanks{The material in this paper was presented in part at the Asilomar Conference on Signals, Systems, and Computers, Pacific Grove, CA, November 2010.}
\thanks{B. C. Jung is with the Department of Information and Communication Engineering, Gyeongsang National University, Tongyeong 650-160, Republic of Korea
(E-mail: bcjung@gnu.ac.kr).}
\thanks{D. Park is with SAIT, Samsung Electronics Co., Ltd., Yongin 446-712, Republic of Korea (E-mail: dohyung22.park@gmail.com).}
\thanks{W.-Y. Shin (corresponding author) is with the School of Engineering and Applied Sciences, Harvard
University, Cambridge, MA 02138 USA
(E-mail:wyshin@seas.harvard.edu).}
        } \maketitle


\markboth{Submitted to IEEE Transactions on Communications} {Jung,
Park, and Shin: Opportunistic Interference Mitigation Achieves
Optimal Degrees-of-Freedom in Wireless Multi-cell Uplink Networks}


\newtheorem{definition}{Definition}
\newtheorem{theorem}{Theorem}
\newtheorem{lemma}{Lemma}
\newtheorem{example}{Example}
\newtheorem{corollary}{Corollary}
\newtheorem{proposition}{Proposition}
\newtheorem{conjecture}{Conjecture}
\newtheorem{remark}{Remark}

\def \diag{\operatornamewithlimits{diag}}
\def \min{\operatornamewithlimits{min}}
\def \max{\operatornamewithlimits{max}}
\def \log{\operatorname{log}}
\def \max{\operatorname{max}}
\def \rank{\operatorname{rank}}
\def \out{\operatorname{out}}
\def \exp{\operatorname{exp}}
\def \arg{\operatorname{arg}}
\def \E{\operatorname{E}}
\def \tr{\operatorname{tr}}
\def \SNR{\operatorname{SNR}}
\def \dB{\operatorname{dB}}
\def \ln{\operatorname{ln}}
\def \bmat{ \begin{bmatrix} }
\def \emat{ \end{bmatrix} }

\def \be {\begin{eqnarray}}
\def \ee {\end{eqnarray}}
\def \ben {\begin{eqnarray*}}
\def \een {\end{eqnarray*}}

\begin{abstract}
We introduce an opportunistic interference mitigation (OIM)
protocol, where a user scheduling strategy is utilized in $K$-cell
uplink networks with time-invariant channel coefficients and base
stations (BSs) having $M$ antennas. Each BS opportunistically
selects a set of users who generate the minimum interference to the
other BSs. Two OIM protocols are shown according to the number $S$
of simultaneously transmitting users per cell: opportunistic
interference nulling (OIN) and opportunistic interference alignment
(OIA). Then, their performance is analyzed in terms of
degrees-of-freedom (DoFs). As our main result, it is shown that $KM$
DoFs are achievable under the OIN protocol with $M$ selected users
per cell, if the total number $N$ of users in a cell scales at least
as $\text{SNR}^{(K-1)M}$. Similarly, it turns out that the OIA
scheme with $S$($<M$) selected users achieves $KS$ DoFs, if $N$
scales faster than $\text{SNR}^{(K-1)S}$. These results indicate
that there exists a trade-off between the achievable DoFs and the
minimum required $N$. By deriving the corresponding upper bound on
the DoFs, it is shown that the OIN scheme is DoF-optimal. Finally,
numerical evaluation, a two-step scheduling method, and the
extension to multi-carrier scenarios are shown.
\end{abstract}

\begin{keywords}
Base station (BS), channel state information, cellular network,
degrees-of-freedom (DoFs), interference, opportunistic interference
alignment (OIA), opportunistic interference mitigation (OIM),
opportunistic interference nulling (OIN), uplink, user scheduling.
\end{keywords}

\newpage


\section{Introduction}

Interference between wireless links has been taken into account as a
critical problem in communication systems. Especially, there exist
three categories of the conventional interference management in
multi-user wireless networks: decoding and cancellation, avoidance
(i.e., orthogonalization), and averaging (or spreading). To consider
both intra-cell and inter-cell interferences of wireless cellular
networks, a simple infinite cellular multiple-access channel (MAC)
model, referred to as the Wyner's model, was characterized and then
its achievable throughput performance was analyzed in~\cite{Wyner1,
Shamai1, Shamai2, Shamai3}. Moreover, joint processing strategy
among multi-cells was developed in a Wyner-like cellular model in
order to efficiently manage the inter-cell
interferences~\cite{Shamai4, Shamai6}. Such cooperation among cells
can be taken into account as another important interference
management scheme. Even if the work in~\cite{Wyner1, Shamai1,
Shamai2, Shamai3, Shamai4, Shamai6} leads to remarkable insight into
complex and analytically intractable practical cellular
environments, the model under consideration is hardly realistic.

Recently, as an alternative approach to show Shannon-theoretic
limits, interference alignment~(IA) was proposed by fundamentally
solving the interference problem when there are two communication
pairs~\cite{MaddahAliMotahariKhandani:08}. It was shown
in~\cite{Jafar_IA_original} that the IA scheme can achieve the
optimal degrees-of-freedom~(DoFs), which are equal to $K/2$, in the
$K$-user interference channel with time-varying channel
coefficients. The basic idea of the scheme is to confine all the
undesired interference from other communication links into a
pre-defined subspace, whose dimension approaches that of the desired
signal space. Hence, it is possible for all users to achieve one
half of the DoFs that we could achieve in the absence of
interference. Since then, interference management schemes based on
IA have been further developed and analyzed in various wireless
network environments: multiple-input multiple-output (MIMO)
interference network~\cite{Jafar_IA_distributed,Jafar_IA_MIMO}, X
network~\cite{Jafar_IA_X_channel,Jafar_Shamai}, and cellular
network~\cite{Tse_IA,MotahariGharanMaddah-AliKhandani,LeeShinClerckx}.
However, the conventional IA
schemes~\cite{Jafar_IA_original,Jafar_IA_MIMO,Ergodic_Viswanath}
require global channel state information (CSI) including the CSI of
other communication links. Furthermore, a huge number of dimensions
based on time/frequency expansion are needed to achieve the optimal
DoFs~\cite{Jafar_IA_original,Jafar_IA_MIMO,Jafar_IA_X_channel,Jafar_Shamai,Tse_IA,Ergodic_Viswanath}.
These constraints need to be relaxed in order to apply IA to more
practical systems. In~\cite{Jafar_IA_distributed}, a distributed IA
scheme was constructed for the MIMO interference channel with
time-invariant coefficients. It requires only local CSI at each node
that can be acquired from all received channel links via pilot
signaling, and thus is more feasible to implement than the original
one~\cite{Jafar_IA_original}. However, a great number of iterations
should be performed until designed transmit/receive beamforming (BF)
vectors converge prior to data transmission.

Now we would like to consider practical wireless uplink networks
with $K$-cells, each of which has $N$ users. IA for $K$-cell uplink
networks was first proposed in~\cite{Tse_IA}, where the interference
from other cells is aligned into a multi-dimensional subspace
instead of one dimension. This scheme also has practical challenges
including a dimension expansion to achieve the optimal DoFs.

In the literature, there are some results on the usefulness of
fading in single-cell downlink broadcast channels, where one can
obtain a multi-user diversity (MUD) gain as the number of mobile
users is sufficiently large: opportunistic
scheduling~\cite{Knopp_Opp}, opportunistic BF~\cite{Viswanath_Opp},
and random BF~\cite{Hassibi_RBF}. More efficient opportunistic
interference management strategy~\cite{Sadjadpour09,Sadjadpour10},
which requires less feedback overhead than that
in~\cite{Hassibi_RBF}, has been developed in broadcast channels,
where similarly as in our study, the minimum number of users needed
for achieving target DoFs has been analyzed.\footnote{Note that the
work in~\cite{Sadjadpour09,Sadjadpour10} was originally conducted in
a single-cell downlink system, but can be extended to multi-cell
downlink environments with a slight modification.} Scenarios
exploiting the MUD gain have also been studied in cooperative
networks by applying an opportunistic two-hop relaying
protocol~\cite{Poor_Opp} and an opportunistic
routing~\cite{Shin_Opp}, and in cognitive radio networks with
opportunistic scheduling~\cite{bcjung_CR, ShenFitz}. In addition,
recent results~\cite{Ergodic_Viswanath,JeonChung} have shown how to
utilize the opportunistic gain when we have a large number of
channel realizations. More specifically, to amplify signals and
cancel interference, the idea of opportunistically pairing
complementary channel instances has been studied in interference
networks~\cite{Ergodic_Viswanath} and multi-hop relay
networks~\cite{JeonChung}. In cognitive radio
environments~\cite{Perlaza_OIA,Perlaza_OIA2,ZhangLiang},
opportunistic spectrum sharing was introduced by allowing the
secondary users to share the radio spectrum originally allocated to
the primary users via transmit adaptation in space, time, or
frequency.

In this paper, we introduce an \emph{opportunistic interference
mitigation (OIM)} protocol for wireless multi-cell uplink networks.
The scheme adopts the notion of MUD gain for performing interference
management. The opportunistic user scheduling strategy is presented
in $K$-cell uplink environments with time-invariant channel
coefficients and base stations (BSs) having $M$ receive antennas. In
the proposed OIM scheme, each BS opportunistically selects a set of
users who generate the minimum interference to the other BSs, while
in the conventional opportunistic
algorithms~\cite{Knopp_Opp,Viswanath_Opp,Hassibi_RBF}, users with
the maximum signal strength at the desired BS are selected for data
transmission. Specifically, two OIM protocols are proposed according
to the number $S$ of simultaneously transmitting users per cell:
opportunistic interference nulling~(OIN) and opportunistic
interference alignment~(OIA) protocols. For the OIA scheme, each BS
broadcasts its pre-defined interference direction, e.g., a set of
orthonormal random vectors, to all the users in other cells, whereas
for the OIN scheme, no broadcast is needed at each BS. Each user
computes the amount of its generating interference, affecting the
other BSs, and feeds back it to its home cell BS.

Their performance is then analyzed in terms of achievable DoFs (also
known as capacity pre-log factor or multiplexing gain). It is shown
that $KM$ DoFs are achievable under the OIN protocol with $M$
selected users per cell, while the OIA scheme with $S$ selected
users, whose number is smaller than $M$, achieves $KS$ DoFs. As our
main result, we analyze the scaling condition between the number $N$
of per-cell users the received signal-to-noise ratio (SNR) under
which our achievability result holds in $K$-cell networks, each of
which has $N$ users. More specifically, we show that the
aforementioned DoFs are achieved asymptotically, provided that $N$
scales faster than $\text{SNR}^{(K-1)M}$ and $\text{SNR}^{(K-1)S}$
for the OIN and OIA protocols, respectively. From the result, it is
seen that there exists a fundamental trade-off between the
achievable DoFs and the minimum required number $N$ of users per
cell, based on the two proposed schemes. In addition, we derive an
upper bound on the DoFs in $K$-cell uplink networks. It is shown
that the upper bound always approaches $KM$ regardless of $N$ and
thus the OIN scheme achieves the optimal DoFs asymptotically with
the help of the opportunism.

Some important aspects are discussed as follows. To validate the OIA
scheme, computer simulations are performed---the amount of
interference leakage is evaluated as
in~\cite{Jafar_IA_distributed,YuSung}. In addition, the conventional
opportunistic mechanism exploiting the MUD gain in the
literature~\cite{Knopp_Opp,Viswanath_Opp,Hassibi_RBF} inspires us to
introduce a two-step scheduling strategy with a slight modification.
We show that a logarithmic gain can further be obtained, similarly
as in~\cite{Knopp_Opp,Viswanath_Opp,Hassibi_RBF}, while the full
DoFs are maintained. Extension to multi-carrier systems of our
achievability result is also taken into account. Finally, the
proposed scheme is also compared with the existing methods which can
also asymptotically achieve the optimal DoFs in cellular uplink
networks.

As in~\cite{Jafar_IA_distributed}, the OIM protocol basically
operates with local CSI and no time/frequency expansion, thereby
resulting in easier implementation. No iteration is also needed
prior to data transmission. The scheme thus operates as a
decentralized manner which does not involve joint processing among
all communication links.

The rest of this paper is organized as follows. In
Section~\ref{sect:System_model}, we introduce the system and channel
models. In Section~\ref{sect:Proposed}, the OIM technique is
proposed for cellular networks and its achievability in terms of
DoFs is also analyzed. Section~\ref{sect:Upperbound_DoF} shows an
upper bound on the DoFs. Numerical evaluation, the two-step
scheduling method, extension to multi-carrier scenarios, and
comparison with the existing methods are shown in
Section~\ref{sect:Numerical_EX}. Finally, we summarize the paper
with some concluding remark in Section~\ref{sect:Conclusion}.

Throughout this paper, the superscripts $T$, $H$, and $\dagger$
denote the transpose, conjugate transpose, and pseudo-inverse,
respectively, of a matrix (or a vector). $\mathbb{C}$, $\|\cdot\|$,
$\mathbf{I}_n$, $\lambda_{\min}(\cdot)$, $\mathbb{E}[\cdot]$, and
$\text{diag}(\cdot)$ indicate the field of complex numbers,
$L_2$-norm of a vector, the identity matrix of size $n\times n$, the
smallest eigenvalue of a matrix, and the statistical expectation,
and the vector consisting of the diagonal elements of a matrix,
respectively.

\section{\label{sect:System_model} System and Channel Models}

Consider the interfering MAC (IMAC) model in~\cite{Tse_IA}, which is
one of multi-cell uplink scenarios, to describe practical cellular
networks.
 As illustrated in Fig.~\ref{FIG:system}, there are
multiple cells, each of which has multiple mobile users. The example
for $K=2$, $N=3$, and $M=2$ is shown in Fig.~\ref{FIG:system}. Under
the model, each BS is interested only in traffic demands of users in
the corresponding cell. Suppose that there are $K$ cells and there
are $N$ users in a cell. We assume that each user is equipped with a
single transmit antenna and each cell is covered by one BS with $M$
receive antennas. The channel in a single-cell can then be regarded
as the single-input multiple-output~(SIMO) MAC. If $N$ is much
greater than $M$, then it is possible to exploit the channel
randomness and thus to obtain the opportunistic gain in multi-user
environments.

The term $\mathbf{h}_{i,j}^{(k)} \in \mathbb{C}^{M \times 1}$
denotes the channel vector between user $j$ in the $k$-th cell and
BS $i$, where $j\in\{1,\cdots,N\}$ and $i,k\in\{1,\cdots,K\}$. The
channel is assumed to be Rayleigh, whose elements have zero-mean and
unit variance, and to be independent across different $i$, $j$, and
$k$. We assume a block-fading model, i.e., the channel vectors are
constant during one block (e.g., frame) and changes to a new
independent value for every block. The receive signal vector
$\mathbf{y}_i\in \mathbb{C}^{M \times 1}$ at BS $i$ is given by
\begin{eqnarray}\label{eq.receive_vector_BS_i}
\mathbf{y}_i &=& \sum_{j=1}^{S} \mathbf{h}_{i,j}^{(i)}x_{j}^{(i)} +
\sum_{k=1, k \neq i}^{K} \sum_{n=1}^{S}
\mathbf{h}_{i,n}^{(k)}x_{n}^{(k)} + \mathbf{z}_i,
\end{eqnarray}
where $x_{j}^{(i)}$ is the transmit symbol of user $j$ in the $i$-th
cell and $S$ represents the number of users transmitting data
simultaneously in each cell for $S\in\{1,\cdots,M\}$. The received
signal $\mathbf{y}_i$ at BS $i$ is corrupted by the independently
identically distributed (i.i.d.) and circularly symmetric complex
additive white Gaussian noise~(AWGN) vector $\mathbf{z}_i \in
\mathbb{C}^{M \times 1}$ whose elements have zero-mean and variance
$N_0$. We assume that each user has an average transmit power
constraint $\mathbb{E}\left[ \left|x_{j}^{(i)}\right|^2 \right] \leq
P$. Then, the received SNR at each BS is expressed as a function of
$P$ and $N_0$, which depends on the decoding process at the receiver
side. In this work, we take into account a simple zero-forcing (ZF)
receiver based on the channel vectors between the BS and its
selected home cell users, which will be discussed in detail in
Section~\ref{sect:OIM}.

\section{\label{sect:Proposed}  Achievability Result}

We propose the following two OIM protocols: OIN and OIA protocols.
Then, their performance is analyzed in terms of achievable DoFs.

\subsection{OIM in $K$-cell Uplink Networks} \label{sect:OIM}

We mainly focus on the case for $SK>M$, since otherwise we can
simply achieve the maximum DoFs by applying the conventional ZF
receiver (at BS $i\in\{1,\cdots,K\}$) based on the following channel
transfer matrix
\begin{eqnarray}
\bmat\mathbf{h}_{1,1}^{(i)}&\cdots&\mathbf{h}_{1,S}^{(i)}&\cdots&\mathbf{h}_{K,1}^{(i)}&\cdots&\mathbf{h}_{K,S}^{(i)}\emat.
\nonumber
\end{eqnarray}

\subsubsection{OIN Protocol}

We first introduce an OIN protocol with which $M$ selected users in
a cell transmit their data simultaneously, i.e., the case where
$S=M$. It is possible for user $j$ in the $i$-th cell to obtain all
the cross-channel vectors $\mathbf{h}_{k,j}^{(i)}$ by utilizing a
pilot signaling sent from other cell BSs, where
$j\in\{1,\cdots,N\}$, $i\in\{1,\cdots,K\}$, and
$k\in\{1,\cdots,i-1,i+1,\cdots,K\}$.

We now examine how much the cross-channels of selected users are in
deep fade by computing the following value $L_{k,j}^i$:
\begin{equation}
L_{k,j}^i=\left\|\mathbf{h}_{k,j}^{(i)}\right\|^2,
\label{eq.leakage_of_IN}
\end{equation}
which is called {\em leakage of interference (LIF)}, for
$k\in\{1,\cdots,i-1,i+1,\cdots,K\}$. For user $j$ in the $i$-th
cell, the user scheduling metric $L_j^i$ is given by
\begin{eqnarray} \label{eq.leakage_of_IA_sum}
L_{j}^{i} = \sum_k L_{k,j}^{i}
\end{eqnarray}
for $k\in\{1,\cdots,i-1,i+1,\cdots,K\}$. After computing the metric
representing the total sum of $K-1$ LIF values in
(\ref{eq.leakage_of_IA_sum}), each user feeds back the value to its
home cell BS $i$.\footnote{An opportunistic feedback strategy can be
adopted in order to reduce the amount of feedback overhead without
any performance loss, similarly as in MIMO broadcast
channels~\cite{TangHeathChoYun}, even if the details are not shown
in this paper.} Thereafter, BS $i$ selects a set $\{\pi_i(1),
\ldots, \pi_i(M)\}$ of $M$ users who feed back the values up to the
$M$-th smallest one in (\ref{eq.leakage_of_IA_sum}), where
$\pi_i(j)$ denotes the index of users in cell $i$ whose value is the
$j$-th smallest one. The selected $M$ users in each cell start to
transmit their data packets.

At the receiver side, each BS performs a simple ZF filtering based
on intra-cell channel vectors to detect the signal from its home
cell users, which is sufficient to capture the full DoFs in our
model. The resulting signal (symbol), postprocessed by ZF matrix
$\mathbf{G}_i\in\mathbb{C}^{M\times M}$ at BS $i$, is then given by
\begin{eqnarray} \label{eq.x_symbol}
\bmat \hat{x}_1^{(i)}&\cdots& \hat{x}_{M}^{(i)}
\emat^{T}=\mathbf{G}_i \mathbf{y}_i,
\end{eqnarray}
where
\begin{eqnarray}
\mathbf{G}_i\!\!\!\!\!\!\!&&= \bmat \bar{\mathbf{g}}_1^{(i)} ~
\cdots ~ \bar{\mathbf{g}}_M^{(i)} \emat^T \nonumber\\ &&= \bmat
\mathbf{h}_{i,1}^{(i)} ~ \cdots ~ \mathbf{h}_{i,M}^{(i)}
\emat^{\dagger} \nonumber
\end{eqnarray}
and $\bar{\mathbf{g}}_m^{(i)}\in\mathbb{C}^{M\times1}$
($m=1,\cdots,M$) is the ZF column vector.

\subsubsection{OIA Protocol}

The fact that the OIN scheme needs a great number of per-cell users
motivates the introduction of an OIA protocol in which $S$
transmitting users are selected in each cell for
$S\in\{1,\cdots,M-1\}$. The OIA scheme is now described as follows.
First, BS $i$ in the $i$-th cell generates a set of orthonormal
random vectors $\mathbf{v}_m^{(i)} \in \mathbb{C}^{M \times 1}$ for
all $m=1,\cdots,M-S$ and $i=1,\cdots,K$, where $\mathbf{v}_m^{(i)}$
corresponds to its pre-defined interference direction, and then
broadcasts the random vectors to all the users in other
cells.\footnote{Alternatively, a set of vectors can be generated
with prior knowledge in a pseudo-random manner, and thus can be
acquired by all users before data transmission without any signaling
overhead.} That is, the interference subspace is broadcasted. If
$m_1=m_2$, then $\mathbf{v}_{m_1}^{(i) H}\mathbf{v}_{m_2}^{(i)}=1$
for $m_1,m_2\in\{1,\cdots,M-1\}$. Otherwise, it follows that
$\mathbf{v}_{m_1}^{(i) H}\mathbf{v}_{m_2}^{(i)}=0$. For example, if
$M-S$ is set to 1, i.e., single interference dimension is used, then
$M-1$ users in a cell are selected to transmit their data packets
simultaneously. This can be easily extended to the case where a
multi-dimensional subspace is allowed for IA (e.g., $M-S\ge2$).

With this scheme, it is important to see how closely the channels of
selected users are aligned with the {\em span} of broadcasted
interference vectors. To be specific, let $\{{\bf
u}_1^{(i)},\cdots,{\bf u}_{S}^{(i)}\}$ denote an orthonormal basis
for the null space $U^{(i)}$ (i.e., kernel) of the interference
subspace. User $j\in\{1,\cdots,N\}$ in the $i$-th cell then computes
the orthogonal projection onto $U^{(k)}$ of its channel vector
$\mathbf{h}_{k,j}^{(i)}$, which is given by
\begin{equation}
\text{Proj}_{U^{(k)}}\left(\mathbf{h}_{k,j}^{(i)}\right)=\sum_{m=1}^{S}\left({\bf
u}_m^{(k)H}\mathbf{h}_{k,j}^{(i)}\right){\bf u}_m^{(k)}, \nonumber
\end{equation}
and the value
\begin{eqnarray}\label{eq.leakage_of_IA}
L_{k,j}^{i}=
\left\|\text{Proj}_{U^{(k)}}\left(\mathbf{h}_{k,j}^{(i)}\right)\right\|^2,
\end{eqnarray}
which can be interpreted as the LIF in the OIA scheme, for
$k\in\{1,\cdots,i-1,i+1,\cdots,K\}$. For example, if the LIF of a
user is given by $0$ for a certain another BS
$k\in\{1,\cdots,i-1,i+1,\cdots,K\}$, then it indicates that the
user's channel vectors are perfectly aligned to the interference
direction of BS $k$ and the user's signal does not interfere with
signal detection at the BS. For user $j$ in the $i$-th cell, the
user scheduling metric $L_{j}^{i}$ is finally given by
(\ref{eq.leakage_of_IA_sum}), as in the OIN protocol. The remaining
scheduling steps are the same as those of OIN except that a set
$\{\pi_i(1), \ldots, \pi_i(S)\}$ of $S$ users is selected at BS $i$
instead of $M$ users.

A ZF filtering at BS $i$ is performed based on both random vectors
$\{\mathbf{v}_1^{(i)},\cdots,\mathbf{v}_{M-S}^{(i)}\}$ and the
intra-cell channel vectors
$\{\mathbf{h}_{i,1}^{(i)},\cdots,\mathbf{h}_{i,S}^{(i)}\}$. Then,
the resulting signal, postprocessed by ZF matrix
$\mathbf{G}_i\in\mathbb{C}^{S\times M}$, is given by
\begin{eqnarray}
\bmat \hat{x}_1^{(i)}&\cdots& \hat{x}_{S}^{(i)}
\emat^{T}=\mathbf{G}_i \mathbf{y}_i, \nonumber
\end{eqnarray}
where
\begin{eqnarray}
\mathbf{G}_i\!\!\!\!\!\!\!&&= \bmat \bar{\mathbf{g}}_1^{(i)} ~
\cdots ~ \bar{\mathbf{g}}_S^{(i)} \emat^T \nonumber\\ &&=
\bmat\mathbf{h}_{i,1}^{(i)}&\cdots&\mathbf{h}_{i,S}^{(i)}\emat^{\dagger}
\nonumber
\end{eqnarray}
and $\bar{\mathbf{g}}_m^{(i)}\in\mathbb{C}^{M\times1}$
($m=1,\cdots,S$) is the ZF column vector.

\subsection{Analysis of Achievable DoFs} \label{sect:DoFs}

In this subsection, we show that the OIM scheme with $S$
simultaneously transmitting users per cell achieves the total number
$KS$ of DoFs asymptotically. The achievability is conditioned by the
scaling behavior between the number $N$ of per-cell users and the
received SNR.

The total number $\mathrm{dof}_{\mathrm{total}}$ of DoFs is defined
as~\cite{ZhengTse}
\begin{align} \label{EQ:DoFdef}
\mathrm{dof}_{\mathrm{total}} &= \sum_{i=1}^K \sum_{j=1}^N d_j^{(i)}
\nonumber\\ &= \sum_{i=1}^K \sum_{j=1}^N \left(\lim_{\mathrm{SNR}
\to \infty} \frac{R_j^{(i)}(\mathrm{SNR})}{\log
\mathrm{SNR}}\right),
\end{align}
where $d_{j}^{(i)}$ and $R_j^{(i)}(\text{SNR})$ denote the DoFs and
the rate, respectively, for the transmission of user
$j\in\{1,\cdots,N\}$ in the $i$-th cell
($i=1,\cdots,K$).\footnote{Especially, the definition of DoFs
associated with the IMAC model was shown
in~\cite{MotahariGharanMaddah-AliKhandani}, and is basically the
same as~(\ref{EQ:DoFdef}).} Note that under the OIM protocol,
$\mathrm{dof}_{\mathrm{total}}$ is then lower-bounded by
\begin{equation} \label{eq.sum-rate}
\mathrm{dof}_{\mathrm{total}} \ge \sum_{i=1}^K \sum_{m=1}^{S}
\left(\lim_{\mathrm{SNR} \to \infty}\frac{ \log \left( 1 +
\mathrm{SINR}_{i,m} \right) }{\log \mathrm{SNR}}\right),
\end{equation}
where $\mathrm{SINR}_{i,m}$ denotes the
signal-to-interference-and-noise ratio (SINR) for the desired stream
$m\in\{1,\cdots,S\}$ at the receiver (BS) in the $i$-th cell and is
represented by
\begin{align}
\mathrm{SINR}_{i,m}&=\frac{\left|\bar{\mathbf{g}}_m^{(i) H}
\mathbf{h}_{i,\pi_i(m)}^{(i)} \right|^2  \mathrm{SNR}}{1 +
\sum_{k=1, k \neq i}^K \sum_{j=1}^{S} \left|\bar{\mathbf{g}}_m^{(i)
H} \mathbf{h}_{i,\pi_k(j)}^{(k)} \right|^2 \mathrm{SNR}} \nonumber\\
&\geq \frac{\left|\bar{\mathbf{g}}_m^{(i) H}
\mathbf{h}_{i,\pi_i(m)}^{(i)} \right|^2  \mathrm{SNR}}{1 +
\sum_{k=1, k \neq i}^K \sum_{j=1}^{S} \left\|
\bar{\mathbf{g}}_m^{(i) H} \right\|^2L_{i,\pi_k(j)}^{k} \mathrm{SNR}} \nonumber \\
&=\frac{ \frac{\left|\bar{\mathbf{g}} _m^{(i) H}
\mathbf{h}_{i,\pi_i(m)}^{(i)} \right|^2}{\left\|
\bar{\mathbf{g}}_m^{(i) H} \right\|^2} \mathrm{SNR}}{1 + \sum_{k=1,
k \neq i}^K \sum_{j=1}^{S} L_{i,\pi_k(j)}^{k} \mathrm{SNR}},
\label{eq.SINRim}
\end{align}
where $L_{i,\pi_k(j)}^{k}$ is given by (\ref{eq.leakage_of_IN}) and
(\ref{eq.leakage_of_IA}) when $S=M$ and $S\in\{1,\cdots,M-1\}$,
respectively. Here, the inequality holds due to the Cauchy-Schwarz
inequality. Now our focus is to characterize the LIF
$L_{i,\pi_k(j)}^{k}$ in order to quantify the achievable total DoFs
$\mathrm{dof}_{\mathrm{total}}$. Since the $M$-dimensional SIMO
channel vector $\mathbf{h}_{i,\pi_k(j)}^{(k)}$ is isotropically
distributed, the user scheduling metric $L^i_j$, representing the
total sum of $K-1$ LIF values, follows the chi-square distribution
with $2(K-1)S$ degrees of freedom for any $i=1,\cdots, K$ and $j =
1,2,\ldots,N$. The cumulative distribution function (cdf) $F_L(l)$
of the metric $L^i_j$ is given by
\begin{align} \label{EQ:F_L}
F_L(l) = \frac{\gamma{((K-1)S,l/2)}}{\Gamma((K-1)S)},
\end{align}
where $\Gamma(z) = \int_0^{\infty} t^{z-1}e^{-t}dt$ is the Gamma
function and $\gamma(z,x) = \int_0^{x} t^{z-1}e^{-t}dt$ is the lower
incomplete Gamma function. We start from the following lemma.

\begin{lemma}\label{lem:gammafunc}
For any $0 \leq l < 2$, the cdf $F_L(l)$ of the metric $L^i_j$ in
(\ref{eq.leakage_of_IA_sum}) is lower- and upper-bounded by
\begin{align}
C_1 l^{(K-1)S} \leq F_L(l) \leq C_2 l^{(K-1)S}, \label{eq:lemma1}
\end{align}
where
\begin{align}
C_1 &= \frac{e^{-1}2^{-(K-1)S}}{(K-1)S \cdot
\Gamma\left((K-1)S\right)}, \nonumber
\end{align}
\begin{align}
C_2 &=
\frac{2\cdot2^{-(K-1)S}}{(K-1)S\cdot\Gamma\left((K-1)S\right)},
\nonumber
\end{align}
and $\Gamma(z)$ is the Gamma function.
\end{lemma}

The proof of this lemma is presented in Appendix~\ref{PF:gammafunc}.
It is now possible to derive the achievable DoFs for $K$-cell uplink
networks using the OIM protocol.

\begin{theorem} \label{THM:DOF_lower}
Suppose that the OIM scheme with $S$ simultaneously transmitting
users in a cell is used in the IMAC model. Then,
\begin{equation}
\mathrm{dof}_{\mathrm{total}}\ge KS \label{eq.achievableDoF}
\end{equation}
is achievable with high probability (whp), if $N = \omega \left(
\mathrm{SNR}^{(K-1)S} \right)$, where
$S=\{1,\cdots,M\}$.\footnote{We use the following notations: i)
$f(x)=O(g(x))$ means that there exist constants $C$ and $c$ such
that $f(x)\le Cg(x)$ for all $x>c$. ii) $f(x)=\omega(g(x))$ means
that
$\underset{x\rightarrow\infty}\lim\frac{g(x)}{f(x)}=0$~\cite{Knuth}.}
\end{theorem}

\begin{proof}
From (\ref{eq.sum-rate}) and (\ref{eq.SINRim}), the OIM scheme
achieves $KS$ DoFs if the value
\begin{equation} \label{EQ:INR}
\sum_{k=1, k \neq i}^K \sum_{j=1}^{S} L_{i,\pi_k(j)}^{k}
\mathrm{SNR}
\end{equation}
for all $i \in \{1,2,\ldots,K\}$ and $m \in \{1,2,\ldots, S\}$ is
smaller than or equal to some constant $\epsilon>0$ independent of
SNR. The number $\mathrm{dof}_{\mathrm{total}}$ of DoFs is
lower-bounded by
\begin{align}
\mathrm{dof}_{\mathrm{total}} \geq P_{\mathrm{OIM}}KS, \nonumber
\end{align}
which holds since $KS$ DoFs are achieved for a fraction
$P_{\mathrm{OIM}}$ of the time, from the fact that
$\mathrm{SINR}_{i,m}=\Omega(\mathrm{SNR})$ with probability
$P_{\mathrm{OIM}}$, where
\begin{align}
P_{\mathrm{OIM}} &=  \lim_{\mathrm{SNR} \to \infty} \mathrm{Pr}
\left\{\sum_{k=1, k \neq i}^K \sum_{j=1}^{S} L_{i,\pi_k(j)}^{k}
\mathrm{SNR} \leq \epsilon ~ \textrm{for all} ~ i \in
\{1,2,\ldots,K\}, m \in \{1,2,\ldots, S\} \right\}. \nonumber
\end{align}

We now examine the scaling condition such that $P_{\mathrm{OIM}}$
converges to one whp. For a constant $\epsilon>0$, we have
\begin{align}
P_{\mathrm{OIM}} &\geq \lim_{\mathrm{SNR} \to \infty} \mathrm{Pr}
\left\{ \sum_{i=1}^K \sum_{m=1}^S \sum_{k=1, k \neq i}^K
\sum_{j=1}^{S} L_{i,\pi_k(j)}^{k} \mathrm{SNR} \leq \epsilon \right\} \nonumber\\
&\geq \lim_{\mathrm{SNR} \to \infty} \mathrm{Pr} \left\{ S\sum_{k=1}^K \sum_{j=1}^{S} L^{k}_{\pi_k(j)} \leq \epsilon  \mathrm{SNR}^{-1} \right\} \nonumber \\
&\geq \lim_{\mathrm{SNR} \to \infty}\mathrm{Pr} \left\{ L^k_{\pi_k
(S)}
\leq \frac{\epsilon \mathrm{SNR}^{-1}}{KS^2} ~ \textrm{for all} ~ k \in \{1,\ldots,K\} \right\} \nonumber \\
&= \lim_{\mathrm{SNR} \to \infty}\left(\mathrm{Pr} \left\{
L^1_{\pi_1 (S)} \leq \frac{\epsilon \mathrm{SNR}^{-1}}{KS^2}
\right\}\right)^K, \label{eq:lowerbound}
\end{align}
where the last equality holds from the fact that if $i_1 \neq i_2$,
then $L^{i_1}_j$ and $L^{i_2}_j$ are given by a function of
different random vectors, and thus are independent of each other.
Then, (\ref{eq:lowerbound}) can further be lower-bounded by using
\begin{align}
& \lim_{\mathrm{SNR} \to \infty}\mathrm{Pr} \left\{ L^1_{\pi_1 (S)}
\leq \frac{\epsilon \mathrm{SNR}^{-1}}{KS^2} \right\} \nonumber \\
&= 1 - \lim_{\mathrm{SNR} \to \infty} \sum_{i=0}^{S-1} \left( N
\atop i \right) F_L \left( \frac{\epsilon \mathrm{SNR}^{-1}}{KS^2}
\right)^i  \left( 1-F_L
\left( \frac{\epsilon \mathrm{SNR}^{-1}}{KS^2} \right) \right)^{N-i} \nonumber\\
&\geq 1 - \lim_{\mathrm{SNR} \to \infty} \sum_{i=0}^{S-1} \frac{
\left( N  C_2  \left( \frac{\epsilon}{2KS^2} \right)^{(K-1)S}
\mathrm{SNR}^{-(K-1)S} \right)^i  \left( 1- C_1  \left(
\frac{\epsilon}{2KS^2} \right)^{(K-1)S} \mathrm{SNR}^{-(K-1)S}
\right)^{N}}{\left( 1 - C_2  \left( \frac{\epsilon}{2KS^2}
\right)^{(K-1)S} \mathrm{SNR}^{-(K-1)S} \right)^i},
\label{eq:lowerbound2} \nonumber
\end{align}
where the inequality holds due to Lemma \ref{lem:gammafunc}. If $N =
\omega \left( \mathrm{SNR}^{(K-1)S} \right)$, then the value
\begin{equation}
\left( N  C_2  \left( \frac{\epsilon}{2KS^2} \right)^{(K-1)S}
\mathrm{SNR}^{-(K-1)S} \right)^i  \left( 1- C_1  \left(
\frac{\epsilon}{2KS^2} \right)^{(K-1)S} \mathrm{SNR}^{-(K-1)S}
\right)^{N} \label{eq.twoterms}
\end{equation}
converges to zero for all $i=0,\cdots, S-1$, because in
(\ref{eq.twoterms}), the second term decays exponentially with
increasing SNR while the first term increases rather polynomially.
The lower bound in (\ref{eq:lowerbound}) thus converges to one.

As a consequence, our result indicates that the term $\sum_{k=1, k
\neq i}^K \sum_{j=1}^{S} L_{i,\pi_k(j)}^{k}$ scales as $O \left(
\mathrm{SNR}^{-1} \right)$ whp if $N = \omega \left(
\mathrm{SNR}^{(K-1)S} \right)$. This further implies that for the
decoded symbol $\hat{x}^{(i)}_m$, the value in (\ref{EQ:INR}) is
smaller than or equal to $\epsilon$ with probability
$P_{\mathrm{OIM}}$, approaching one, as the received SNR tends to
infinity, where $i\in\{1,\cdots,K\}$ and $m \in\{1,\cdots,S\}$.
Therefore, it follows that $\mathrm{dof}_{\mathrm{total}}\geq KS$ if
$N = \omega \left( \mathrm{SNR}^{(K-1)S} \right)$, which completes
the proof of this theorem.
\end{proof}

From the above theorem, let us show the following interesting
discussion according to the two proposed protocols.

\begin{remark}
It is seen that the asymptotically achievable DoFs are given by $KM$
and $KS$ ($S\in\{1,\cdots,M-1\}$) when the OIN and OIA protocols are
used in $K$-cell uplink networks, respectively. In fact, the OIN
scheme achieves the optimal DoFs, which will be proved in
Section~\ref{sect:Upperbound_DoF} by showing an upper bound on the
DoFs, while it works under the condition that the required number
$N$ of users per cell scales faster than $\mathrm{SNR}^{(K-1)M}$. On
the other hand, the OIA scheme operates with at least
$\mathrm{SNR}^{(K-1)S}$ users per cell, which are surely smaller
than those of the OIN scheme, at the expense of some DoF loss. This
thus gives us a trade-off between the achievable number of DoFs and
the required number $N$ of users in a cell. Note that for the case
where $N$ is not sufficiently large to utilize the OIN scheme, the
OIA scheme can instead be applied in the networks.
\end{remark}

It is now examined how our scheme is fundamentally different from
the existing DoF-optimal
schemes~\cite{Jafar_IA_original,Jafar_IA_MIMO,Jafar_IA_X_channel,Jafar_Shamai,Tse_IA,Ergodic_Viswanath}.

\begin{remark}
As addressed before, the minimum number $N$ of per-cell users needs
to be guaranteed in order that the proposed OIM protocols work
properly even in the time-invariant channel condition without any
dimension expansion. On the other hand,
in~\cite{Jafar_IA_original,Jafar_IA_MIMO,Jafar_IA_X_channel,Jafar_Shamai,Tse_IA,Ergodic_Viswanath},
a huge number of dimensions are required to asymptotically achieve
the optimal DoFs.
\end{remark}

%


\section{Upper Bound for DoFs} \label{sect:Upperbound_DoF}

In this section, to verify the optimality of the proposed OIN
scheme, we derive an upper bound on the DoFs in cellular networks,
especially for the IMAC model shown in Fig. \ref{FIG:system}.
Suppose that $\tilde{N}$ users (i.e., $N$ streams) per cell transmit
their packets simultaneously to the corresponding BS, where
$\tilde{N}\in\{1, 2,\cdots, N\}$.\footnote{Note that $\tilde{N}$ is
different from $S$ in Section~\ref{sect:System_model} since
$\tilde{N}$ can be greater than $M$ in general.} This is a
generalized version of the transmission since it is not
characterized how many users in a cell need to transmit their
packets simultaneously to obtain the optimal DoFs.
An upper bound on the total DoFs for the IMAC model is given in the
following theorem.

\begin{theorem} \label{THM:upper}
For the IMAC model shown in Section~\ref{sect:System_model}, the
total number $\mathrm{dof}_{\mathrm{total}}$ of DoFs is
upper-bounded by
\begin{eqnarray}
\mathrm{dof}_{\mathrm{total}} =
\sum_{i=1}^{K}\sum_{j=1}^{N}d_{j}^{(i)}\le \frac{KNM}{N+1},
\label{EQ:upper}
\end{eqnarray}
where $d_{j}^{(i)}$ denotes the DoFs for the transmission of user
$j$ in the $i$-th cell for $i=1,\cdots,K$ and $j=1,\cdots,N$.
\end{theorem}

The proof of this theorem is presented in Appendix~\ref{PF:upper}.
Note that this upper bound is generally derived regardless of
whether the number $N$ of users per cell tends to infinity or not.
Thus, our converse result always holds for arbitrary $N$, whereas
the scaling condition $N=\omega(\text{SNR}^{(K-1)M})$ is included in
the achievability proof. Now let us turn to examining how the upper
bound is close to the achievable DoFs shown in
Section~\ref{sect:Proposed}.

\begin{remark}
From Theorems~\ref{THM:DOF_lower} and~\ref{THM:upper}, when the OIN
scheme is used (i.e., the case of $S=M$), it is shown that the upper
bound on the DoFs matches the achievable DoFs as long as $N$ scales
faster than $\text{SNR}^{(K-1)M}$. Therefore, the proposed OIN
scheme is optimal in terms on DoFs.
\end{remark}

In addition, a simple upper bound can also be derived in the
following argument.

\begin{remark}
From a genie-aided removal of all the inter-cell interferences, we
obtain $K$ parallel SIMO MAC systems. The number of total DoFs is
thus upper-bounded by $KM$ due to the fact that the number of DoFs
for the SIMO MAC is given by $M$~\cite{Tse_book, Viswanath_Upper}.
It is seen that the upper bound in (\ref{EQ:upper}) approaches $KM$
as the number $N$ of users per cell tends to infinity.
\end{remark}

\section{\label{sect:Numerical_EX} Discussions}

Some important aspects for the proposed scheme are discussed in this
section. We first perform computer simulations to validate the
performance of the proposed OIA scheme in cellular networks. A
two-step user scheduling method is also introduced with a slight
modification, where a logarithmic gain can be obtain. Furthermore,
we show that our achievable scheme can be extended to multi-carrier
systems by executing dimension expansion over the frequency domain.

\subsection{\label{subsect:Simulations} Numerical Evaluation}

The average amount of interference leakage is evaluated as the
number $N$ of users in each cell increases. In our simulation, the
channel vectors in (\ref{eq.receive_vector_BS_i}) are generated
$1\times10^5$ times for each system parameter.

In Fig.~\ref{FIG:leakage_K_2}, The log-log plot of interference
leakage versus $N$ is shown as $N$ increases.\footnote{Even if it
seems unrealistic to have a great number of users in a cell, the
range for parameter $N$ is taken into account to precisely see some
trends of curves varying with $N$.} The interference leakage is
interpreted as the total interference power remaining in each
desired signal space (from the users in other cells) after the ZF
filter is applied, assuming that the received signal power from a
desired transmitter is normalized to 1 in the signal space. This
performance measure enables us to measure the quality of the
proposed OIA scheme, as shown in~\cite{Jafar_IA_distributed,YuSung}.
We now evaluate the interference leakage for various system
parameters. In Fig.~\ref{FIG:leakage_K_2}, the case with $M=8$,
$K=2$, and $SK>M$ is considered, where $S$ denotes the number of
simultaneously transmitting users per cell. It is shown that when
the parameter $S$ varies from 7 to 5, the interference leakage
decreases due to less interferers, which is rather obvious. The
result, illustrated in Fig.~\ref{FIG:leakage_K_2}, indicates that
the interference leakage tends to decrease linearly with $N$, while
the slopes of the curves are almost identical to each other as $N$
increases. It is further seen how many users per cell are required
to guarantee that the interference leakage is less than an
arbitrarily small $\epsilon>0$ for given parameters $M$, $S$, and
$K$.

\subsection{\label{subsect:MUDgain} Two-step OIN Protocol}
The main result of the paper states that the OIN scheme
asymptotically achieves the optimal DoFs in $K$-cell uplink
networks. Users are opportunistically selected in the sense of
confining the generating interference power to other cell BSs within
a constant independent of SNR, while the other opportunistic
algorithms aim to obtain the MUD gain by selecting users with the
maximum channel gain. We now introduce a two-step opportunistic
scheduling method that enables to obtain an additional logarithmic
gain, i.e., power gain, similarly as
in~\cite{Knopp_Opp,Viswanath_Opp,Hassibi_RBF}, as well as the full
DoF gain.

\begin{itemize}
\item {\em Step 1:} For the $i$-th cell, $\tilde{M}$ users are first selected according to the user scheduling metric
$L_j^{i}$ in (\ref{eq.leakage_of_IA_sum}), where
$\tilde{M}=\omega(M)$ and $i=1,\cdots,K$. That is, the parameter
$\tilde{M}$ needs to scale as a certain function of increasing SNR.
\item {\em Step 2:} Among the $\tilde{M}$ users, $M$ users with the desired channel gains up to the
$M$-th largest one are then chosen based on the metric $\|{\bf
h}_{i,\pi_i'(j)}^{(i)}\|^2$, where $\pi_i'(j)$ denotes the index of
users selected in the first step in cell $i$ for
$j=\{1,\cdots,\tilde{M}\}$.
\end{itemize}

From Theorem~\ref{THM:DOF_lower}, it is easily shown that if
$N=\omega(\mathrm{SNR}^{(K-1)\tilde{M}})$, then the interference in
each desired signal space from $\tilde{M}$ selected users per cell
is confined within a constant independent of SNR. Hence, similarly
as in~\cite{Hassibi_RBF}, the received SNR for each symbol would be
boosted by $\log \tilde{M}$ whp, compared to that shown in
(\ref{eq.x_symbol}), under the condition $\tilde{M}=\omega(M)$. As
$\tilde{M}$ scales with SNR (or equivalently $N$), the scaling laws
of the sum-rate in (\ref{eq.sum-rate}) can be obtained with respect
to $\tilde{M}$, and thus the achievable sum-rate scales as
\begin{equation}
KM\log \left(\text{SNR}\log\tilde{M}\right) \nonumber
\end{equation}
whp.\footnote{The pre-log term can be more boosted when $\tilde{M}$
scales exponentially with SNR (or faster), but this infeasible
scaling condition is not a matter of interest in this work.} Hence,
note that the above two-step procedure leads to performance
improvement on the sum-rate (but not on the DoFs).

\subsection{\label{subsect:discussions} Extension to Multi-carrier Systems}
The OIM scheme can easily be applied to multi-carrier systems by
executing dimension expansion over the frequency domain. Let
$N_{\text{sub}}$ denote the total number of subcarriers, which has
no need for tending to infinity. As a single antenna is simply
assumed at each BS in the multi-carrier environment, each user
transmits a data symbol using $N_{\text{sub}}$ frequency subcarriers
and the received signal vector
$\mathbf{y}_{i}\in\mathbb{C}^{N_{\text{sub}}\times1}$ over the
frequency domain at BS $i$ can then be expressed as
\begin{eqnarray} \label{EQ:received_OFDM}
\mathbf{y}_{i} = \sum_{j=1}^{S} \mathbf{H}_{i,j}^{(i)} x_{j}^{(i)} +
\sum_{k=1, k \neq i}^{K-1} \sum_{n=1}^{S}
\mathbf{H}_{i,n}^{(k)}x_{n}^{(k)} + \mathbf{z}_i, \nonumber
\end{eqnarray}
where $\mathbf{H}_{k,j}^{(i)} \in \mathbb{C}^{N_{\text{sub}}\times
1} $ indicates the frequency response of the channel from the $j$-th
user in the $k$-th cell to BS $i$, $\mathbf{z}_{i} \in
\mathbb{C}^{N_{\text{sub}}\times 1}$ is the AWGN vector over the
frequency domain at BS $i$, and $S\in\{1,\cdots,N_{\text{sub}}\}$ is
the number of users transmitting their data simultaneously in each
cell. We assume a rich scattering multipath fading environment and
thus all elements of $\mathbf{H}_{k,j}^{(i)}$ are assumed to be
statistically independent for all $i,k\in\{1,\cdots,K\}$ and
$j\in\{1,\cdots,N\}$.

For the OIN and OIA protocols under the multi-carrier model, the
user scheduling strategy and its achievability result almost follow
the same steps as those shown in Section~\ref{sect:Proposed}. Hence,
we mainly focus on the scenario where a beamforming can also be
performed at the transmitter side along with the user scheduling.

For example, when the OIA scheme is utilized, it is possible for
each user to reduce the amount of interference caused to the BSs in
other cells by generating a beamforming matrix and then adjusting
its vector directions, while no beamforming is available in
Section~\ref{sect:Proposed} since a single transmit antenna is used
at each user. The optimal diagonal weight matrix ${\bf
W}_{j}^{(i)}\in \mathbb{C}^{N_{\text{sub}}\times N_{\text{sub}}}$
can be designed at each user in the sense of minimizing the total
sum of $K-1$ LIF values defined in (\ref{eq.leakage_of_IA}), i.e.,
the metric $L_j^i$:
\begin{eqnarray} \label{EQ:opt_weight_matrix}
&& {\bf W}_{j}^{(i)} = \arg \min_{{\bf W} \in
\mathbb{C}^{N_{\text{sub}}\times N_{\text{sub}}} } \sum_{l=1, k \neq
i}^{K} \left\| \text{Proj}_{U^{(l)}}\left({\bf W} \mathbf{H}_{l,j}^{(i)}\right)\right\|^2 \\
&&\textrm{subject to} \quad \| \text{diag}({\bf W}) \|^2=1,
\nonumber
\end{eqnarray}
where $U^{(l)}$ denotes the null space of the interference subspace
in the $l$-th cell. Note that each user does not need to feed back
its optimal weight matrix in (\ref{EQ:opt_weight_matrix}) to its
home cell BS. Let ${\bf W}_{j, \text{opt}}^{(i)}$ denote the optimal
solution of (\ref{EQ:opt_weight_matrix}). The $j$-th user in the
$i$-th cell then feeds back the following scheduling metric
$\widetilde{L}_{j}^{i}$ that can be computed again by applying the
optimal weight matrix:
\begin{eqnarray} \label{EQ:novel_LIA}
\widetilde{L}_{j}^{i}= \sum_{l=1, k \neq i}^{K} \left\|
\text{Proj}_{U^{(l)}}\left(\mathbf{\tilde{H}}_{l,j}^{(i)}\right)\right\|^2,
\end{eqnarray}
where
\begin{eqnarray}
\mathbf{\tilde{H}}_{l,j}^{(i)} = {\bf W}_{j, \text{opt}}^{(i)}
\mathbf{H}_{l,j}^{(i)}. \nonumber
\end{eqnarray}

Thereafter, BS $i$ selects a set of $S$ users who feed back the
values up to the $S$-th smallest one in (\ref{EQ:novel_LIA}) among
all users in a cell, where $S\in\{1,\cdots,N_{\text{sub}}-1\}$. This
per-user optimization procedure may yield less amount of the LIF at
each BS than that of the conventional approach without beamforming.
In other words, by applying the beamforming design as well as the
user scheduling, the minimum required number $N$ of users per cell
such that a given LIF value is guaranteed may scale slower than
$\text{SNR}^{(K-1)S}$ shown in Theorem~\ref{THM:DOF_lower}, thus
leading to more feasible network realization.

\subsection{\label{subsect:discussions} Comparison with the Existing Methods}
In this subsection, the proposed scheme is compared with the two
existing strategies~\cite{Tse_IA, MotahariGharanMaddah-AliKhandani}
that also achieve the optimal DoFs in $K$-cell uplink networks. We
now focus on the case for $M=1$, i.e., $K$-cell IMAC model with a
single antenna at each BS, as in~\cite{Tse_IA,
MotahariGharanMaddah-AliKhandani}. Under the model, all of the OIN
and two existing IA methods achieve $K$ DoFs asymptotically as the
number $N$ of users in a cell tends to infinity, while their channel
models and (analytical) approaches are quite different from each
other.

Since the two schemes~\cite{Tse_IA,
MotahariGharanMaddah-AliKhandani} are analyzed in a {\em
deterministic} manner, it is possible to achieve a non-zero number
of DoFs, less than $K$, even for finite $N$ (independent of SNR). In
contrast, the achievability result of the OIN scheme is shown based
on a {\em probabilistic} approach, where infinitely many number of
users per cell, which scales faster than $\text{SNR}^{K-1}$, is
needed to guarantee full DoFs without any dimension expansion.

Now let us turn to discussing channel modelings. The subspace-based
IA scheme~\cite{Tse_IA} was introduced in $K$-cell uplink networks
allowing dimension expansion over the frequency domain, where it
requires $(K-1)$-level decomposability of channels at each link
since designing transmit vectors shown in~\cite{Tse_IA} takes
advantage of decomposed channel matrices. Accordingly, single-path
random delay channels are preferable due to the fact that they are
$(K-1)$-level decomposable and thus are convenient to align
interfering signals in practice. If we assume multipath frequency
selective channels, then the whole channel band should be splitted
into multiple sub-bands, each of which needs to be within coherence
bandwidth and to occupy many subcarriers for dimension expansion,
thereby yielding practical challenges. On the other hand, our scheme
works well with rich scattering environments, because it exploits
channel randomness for either nulling or aligning interfering
signals. However, a highly correlated channel among users (e.g.,
relatively poor scattering environment) may result in performance
degradation for the proposed scheme, since it is difficult to select
users such that the sum of LIF values is small enough.
In~\cite{MotahariGharanMaddah-AliKhandani}, another IA scheme, named
as real IA, has been introduced in cellular uplink networks with
time-invariant {\em real} channel coefficients---the IA operation is
conducted in signal scale but not in signal vector space.
Specifically, the strategy exploits the fact that a real line
consists of infinite rational dimensions. Instead, under the {\em
complex} channel environment, a multi-dimensional Euclidean space is
taken into account to align interference in signal vector space, as
shown in the conventional IA
methods~\cite{MaddahAliMotahariKhandani:08,Jafar_IA_original,Jafar_IA_MIMO,Jafar_IA_X_channel,Jafar_Shamai,Tse_IA}.

\section{\label{sect:Conclusion}Conclusion}
Two types of OIM protocols were proposed in wireless $K$-cell uplink
networks, where they do not require the global CSI, infinite
dimension extension, and parameter adjustment through iteration. The
achievable DoFs were then analyzed---the OIM protocol asymptotically
achieves $KS$ DoFs as long as $N$ scales faster than
$\text{SNR}^{(K-1)S}$, where $S\in\{1,\cdots,M\}$. It has been seen
that there exists a trade-off between the achievable DoFs and the
parameter $N$ based on the two OIM schemes. From the result of the
upper bound on the DoFs, it was shown that the OIM protocol with
$S=M$ achieves the optimal DoFs with the help of the MUD gain. In
addition, the two-step scheduling method that can further obtain a
power gain has been shown, and extension to the multi-carrier
systems has been discussed.

\appendix

\section{Appendix}

\subsection{Proof of Lemma~\ref{lem:gammafunc}} \label{PF:gammafunc}
The cdf $F_L(l)$ of the metric $L_j^i$ satisfies the inequality
$\gamma(z,x) \geq \frac{1}{z} x^{z} e^{-1}$ for $z>0$ and $0 \leq x
< 1$ since
\begin{align}
\gamma(z,x) &= \frac{1}{z} x^{z} e^{-x} + \frac{1}{z} \gamma(z+1,x) \nonumber \\
&= \frac{1}{z} x^{z} e^{-x} + \frac{1}{z(z+1)} x^{{z+1}} e^{-x} +
\cdots \nonumber \\
&\geq \frac{1}{z} x^{z} e^{-1}. \nonumber
\end{align}
Similarly, $\gamma(z,x)$ is upper-bounded by $2 z^{-1} x^{z}$ for
$z>0$ and $0 \leq x < 1$ from the fact that
\begin{align}
\gamma(z,x) &= \frac{1}{z} x^{z} e^{-x} + \frac{1}{z} \gamma(z+1,x) \nonumber \\
&\leq \frac{1}{z} x^{z} e^{-x} + \frac{1}{z} x^z e^{-x} \sum_{i=1}^{\infty} \left( \frac{x}{z+1} \right)^i \nonumber \\
&= \left(\frac{1}{z} + \frac{x}{z+1-x} \right) x^z e^{-x} \nonumber \\
&\leq \frac{2}{z} x^z. \nonumber
\end{align}
Applying the above bounds to (\ref{EQ:F_L}), we finally obtain
(\ref{eq:lemma1}), which completes the proof.


\subsection{Proof of Theorem~\ref{THM:upper}} \label{PF:upper}
Although the proof technique is essentially similar to that
of~\cite{Jafar_IA_original, JafarFakhereddin:07}, the whole steps
are shown here for completeness. Let $W_j^{(i)}$ and $R_j^{(i)}$
denote the message and its transmission rate of user $j$ in the
$i$-th cell, respectively. Consider a certain two-cell IMAC model
illustrated in Fig.~\ref{FIG:upper_fig}, where we eliminate messages
$W_j^{(3)} ,W_j^{(4)},\cdots, W_j^{(K)}$ for all
$j\in\{1,\cdots,\tilde{N}\}$ as well as $W_j^{(2)}$ for
$j\in\{2,\cdots,\tilde{N}\}$. We then obtain the following two
equations:
\begin{eqnarray}
\mathbf{y}_1 &=& \sum_{j=1}^{\tilde{N}}
\mathbf{h}_{1,j}^{(1)}x_{j}^{(1)} +
\mathbf{h}_{1,1}^{(2)}x_{1}^{(2)} + \mathbf{z}_1 \nonumber
\end{eqnarray}
and
\begin{eqnarray}
\mathbf{y}_2 &=& \sum_{j=1}^{\tilde{N}}
\mathbf{h}_{2,j}^{(1)}x_{j}^{(1)} +
\mathbf{h}_{2,1}^{(2)}x_{1}^{(2)} + \mathbf{z}_2, \label{EQ:y2_upp}
\end{eqnarray}
which yield
\begin{eqnarray}
\mathbf{y}_2' &=&
\mathbf{h}_{1,1}^{(2)}\left(\mathbf{h}_{2,1}^{(2)\dagger}\mathbf{h}_{2,1}^{(2)}\right)^{-1}\mathbf{h}_{2,1}^{(2)\dagger}\sum_{j=1}^{\tilde{N}}
\mathbf{h}_{2,j}^{(1)}x_{j}^{(1)} +
\mathbf{h}_{1,1}^{(2)}x_{1}^{(2)} + \mathbf{z}_2' \nonumber
\end{eqnarray}
after multiplying some channel matrices at both sides of
(\ref{EQ:y2_upp}), where
\begin{eqnarray} \nonumber
\mathbf{z}_2'\sim \mathcal{CN}\left(\mathbf{0},N_0
\left(\mathbf{h}_{1,1}^{(2)}\left(\mathbf{h}_{2,1}^{(2)\dagger}\mathbf{h}_{2,1}^{(2)}\right)^{-1}\mathbf{h}_{2,1}^{(2)\dagger}\right)\left(\mathbf{h}_{1,1}^{(2)}\left(\mathbf{h}_{2,1}^{(2)\dagger}\mathbf{h}_{2,1}^{(2)}\right)^{-1}\mathbf{h}_{2,1}^{(2)\dagger}\right)^{\dagger}\right).
\end{eqnarray}
Suppose that $\mathbf{z}_1=\bar{\mathbf{z}}+\bar{\mathbf{z}}_1$ and
$\mathbf{z}_2'=\bar{\mathbf{z}}+\bar{\mathbf{z}}_2$, where
\begin{eqnarray} \nonumber
\bar{\mathbf{z}}\sim \mathcal{CN}\left(\mathbf{0},\alpha
N_0\mathbf{I}_M\right),
\end{eqnarray}
\begin{eqnarray} \nonumber
\bar{\mathbf{z}}_1\sim
\mathcal{CN}\left(\mathbf{0},(1-\alpha)N_0\mathbf{I}_M\right),
\end{eqnarray}
and
\begin{eqnarray} \nonumber
\bar{\mathbf{z}}_2\sim \mathcal{CN}\left(\mathbf{0},N_0
\left(\mathbf{h}_{1,1}^{(2)}\left(\mathbf{h}_{2,1}^{(2)\dagger}\mathbf{h}_{2,1}^{(2)}\right)^{-1}\mathbf{h}_{2,1}^{(2)\dagger}\right)\left(\mathbf{h}_{1,1}^{(2)}\left(\mathbf{h}_{2,1}^{(2)\dagger}\mathbf{h}_{2,1}^{(2)}\right)^{-1}\mathbf{h}_{2,1}^{(2)\dagger}\right)^{\dagger}-\alpha
N_0\mathbf{I}_M\right).
\end{eqnarray}
Here, $\alpha$ is given by
\begin{eqnarray} \nonumber
\alpha=\min\left(1,\lambda_{\min}\left(\left(\mathbf{h}_{1,1}^{(2)}\left(\mathbf{h}_{2,1}^{(2)\dagger}\mathbf{h}_{2,1}^{(2)}\right)^{-1}\mathbf{h}_{2,1}^{(2)\dagger}\right)\left(\mathbf{h}_{1,1}^{(2)}\left(\mathbf{h}_{2,1}^{(2)\dagger}\mathbf{h}_{2,1}^{(2)}\right)^{-1}\mathbf{h}_{2,1}^{(2)\dagger}\right)^{\dagger}\right)\right).
\end{eqnarray}

Then by using Fano's inequality~\cite{CoverThomas:91}, we have
\begin{eqnarray}
\sum_{j=1}^{\tilde{N}}R_{j}^{(1)}+R_1^{(2)} \!\!\!\!\!\!\!&&\le
I\left(W_1^{(1)},\cdots,W_{\tilde{N}}^{(1)};
\mathbf{y}_1\right)+I\left(W_1^{(2)}; \mathbf{y}_2\right)+\epsilon_0
\nonumber\\ &&=I\left(W_1^{(1)},\cdots,W_{\tilde{N}}^{(1)};
\mathbf{y}_1\right)+I\left(W_1^{(2)};
\mathbf{y}_2'\right)+\epsilon_0 \nonumber\\ && \le
I\left(W_1^{(1)},\cdots,W_{\tilde{N}}^{(1)}; \mathbf{y}_1\right)
\nonumber\\ &&\quad+I\biggl(W_1^{(2)};
\mathbf{h}_{1,1}^{(2)}\left(\mathbf{h}_{2,1}^{(2)\dagger}\mathbf{h}_{2,1}^{(2)}\right)^{-1}\mathbf{h}_{2,1}^{(2)\dagger}\sum_{j=1}^{\tilde{N}}
\mathbf{h}_{2,j}^{(1)}x_{j}^{(1)} +
\mathbf{h}_{1,1}^{(2)}x_{1}^{(2)} + \bar{\mathbf{z}} \nonumber\\
&&\quad\bigg|W_1^{(1)},\cdots,W_{\tilde{N}}^{(1)},x_1^{(1)},\cdots,x_{\tilde{N}}^{(1)}\biggr)+\epsilon_0
\nonumber\\&&=I\left(W_1^{(1)},\cdots,W_{\tilde{N}}^{(1)};
\mathbf{y}_1\right) \nonumber\\&& \quad +I\biggl(W_1^{(2)};
\mathbf{h}_{1,1}^{(2)}x_{1}^{(2)} +
\bar{\mathbf{z}}\left|W_1^{(1)},\cdots,W_{\tilde{N}}^{(1)},x_1^{(1)},\cdots,x_{\tilde{N}}^{(1)}\right)+\epsilon_0
\nonumber\\ && \le I\left(W_1^{(1)},\cdots,W_{\tilde{N}}^{(1)};
\sum_{j=1}^{\tilde{N}} \mathbf{h}_{1,j}^{(1)}x_{j}^{(1)} +
\mathbf{h}_{1,1}^{(2)}x_{1}^{(2)} + \bar{\mathbf{z}}\right)
\nonumber\\&& \quad+ I\biggl(W_1^{(2)}; \sum_{j=1}^{\tilde{N}}
\mathbf{h}_{1,j}^{(1)}x_{j}^{(1)} +
\mathbf{h}_{1,1}^{(2)}x_{1}^{(2)} +\bar{\mathbf{z}} \nonumber\\&&
\quad
\bigg|W_1^{(1)},\cdots,W_{\tilde{N}}^{(1)},x_1^{(1)},\cdots,x_{\tilde{N}}^{(1)}\biggr)+\epsilon_0
\nonumber\\&&=
I\left(W_1^{(1)},\cdots,W_{\tilde{N}}^{(1)},W_1^{(2)};
\sum_{j=1}^{\tilde{N}} \mathbf{h}_{1,j}^{(1)}x_{j}^{(1)} +
\mathbf{h}_{1,1}^{(2)}x_{1}^{(2)} +
\bar{\mathbf{z}}\right)+\epsilon_0 \label{EQ:rate}
\end{eqnarray}
for an arbitrarily small $\epsilon_0>0$, where the second and third
inequalities come from reducing noise variance. The right-hand-side
of (\ref{EQ:rate}) represents the sum capacity of a MAC with an $M$
antenna receiver and $\tilde{N}$ single-antenna transmitters, and
thus if $\tilde{N}\ge M$, then the number of DoFs for the MAC is
given by $M$~\cite{Tse_book, Viswanath_Upper}. Hence, simply
assuming $\tilde{N}=N$, we obtain the following upper bounds:
\begin{eqnarray}
\sum_{j=1}^{N}R_{j}^{(1)}+R_1^{(2)}\le
M\log\text{SNR}+o\left(\log\text{SNR}\right) \nonumber
\end{eqnarray}
and
\begin{eqnarray}
\sum_{j=1}^{N}d_{j}^{(1)}+d_1^{(2)}\le M. \nonumber
\end{eqnarray}
Similarly, for any $k\in\{1,2,\cdots,N\}$, we obtain
\begin{eqnarray}
\sum_{j=1}^{N}d_{j}^{(1)}+d_k^{(2)}\le M   \label{EQ:case1_upp}
\end{eqnarray}
and
\begin{eqnarray}
d_{k}^{(1)}+\sum_{j=1}^{N}d_j^{(2)}\le M.   \label{EQ:case2_upp}
\end{eqnarray}
Adding up all the possible combinations over $k$ shown in
(\ref{EQ:case1_upp}) and (\ref{EQ:case2_upp}), we finally have
\begin{eqnarray}
\sum_{j=1}^{N}d_{j}^{(i)}\le \frac{NM}{N+1} \nonumber
\end{eqnarray}
at a given cell $i$. Since there are $K$ cells in the IMAC model,
the total number of DoFs is upper-bounded by (\ref{EQ:upper}), which
completes the proof.

%

\begin{center}
Acknowledgement
\end{center}
The authors would like to thank Sae-Young Chung for his helpful
discussions.

\newpage


\begin{figure}[t!]
  \begin{center}
  \leavevmode \epsfxsize=0.63\textwidth   
  \leavevmode 
  \epsffile{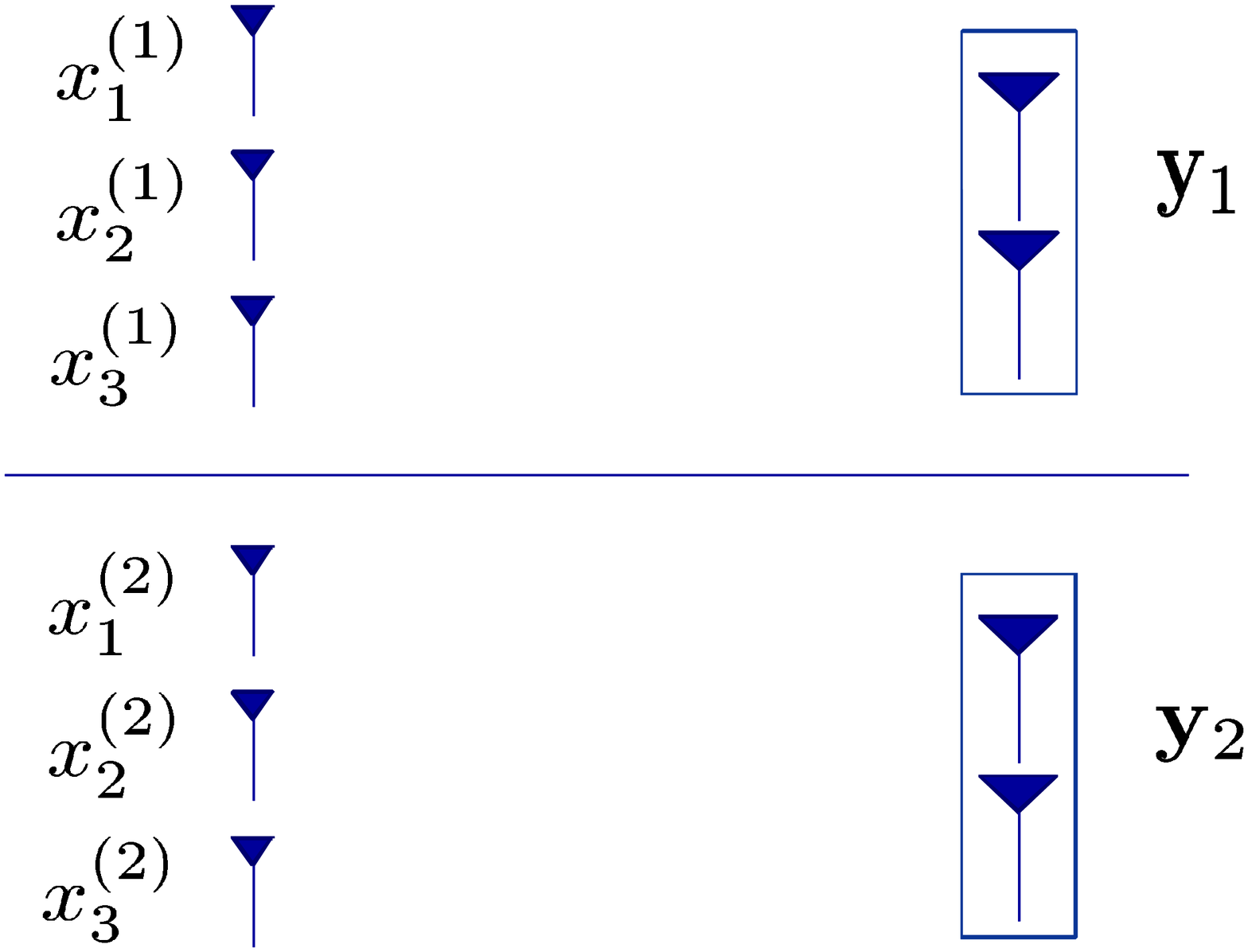}
  \caption{The IMAC model with $K$=2, $N=3$, and $M=2$.}
  \label{FIG:system}
  \end{center}
\end{figure}

\begin{figure}[t!]
  \begin{center}
  \leavevmode \epsfxsize=0.85\textwidth   
  \leavevmode 
  \epsffile{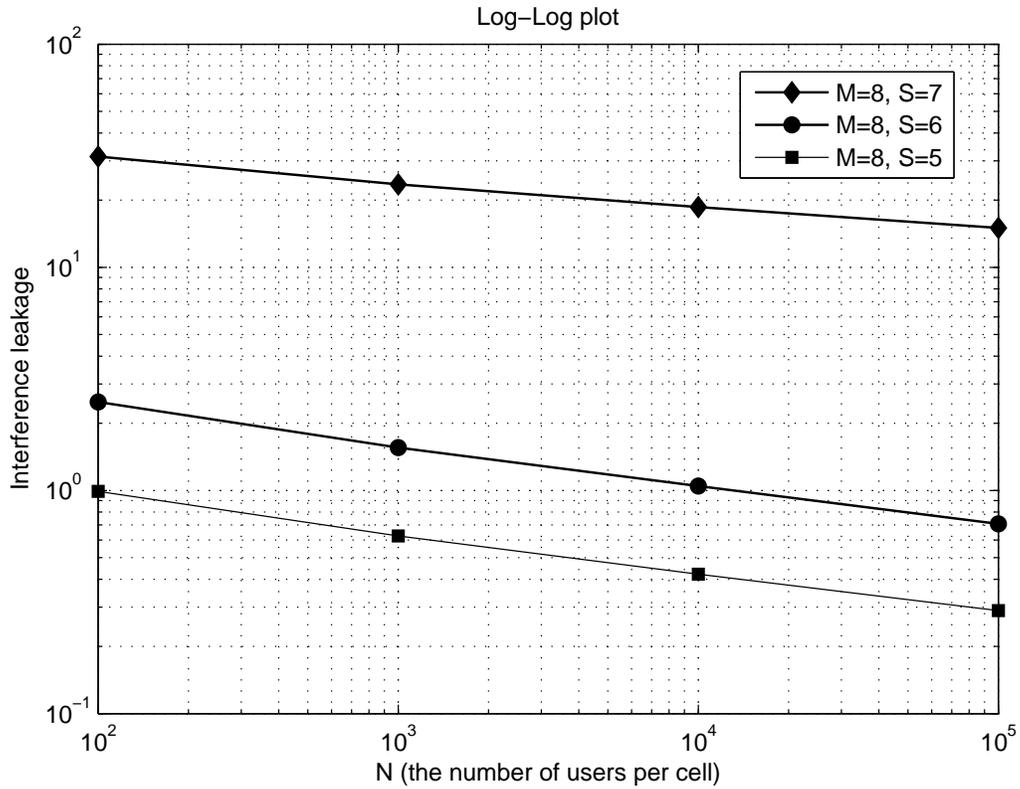}
  \caption{The leakage interference with respect to $N$ for some $S$. The system with $M=8$, $K=2$, and $SK>M$ is considered.}
  \label{FIG:leakage_K_2}
  \end{center}
\end{figure}

\begin{figure}[t!]
\begin{center}
\leavevmode \epsfxsize=0.65 \textwidth   
\leavevmode 
\epsffile{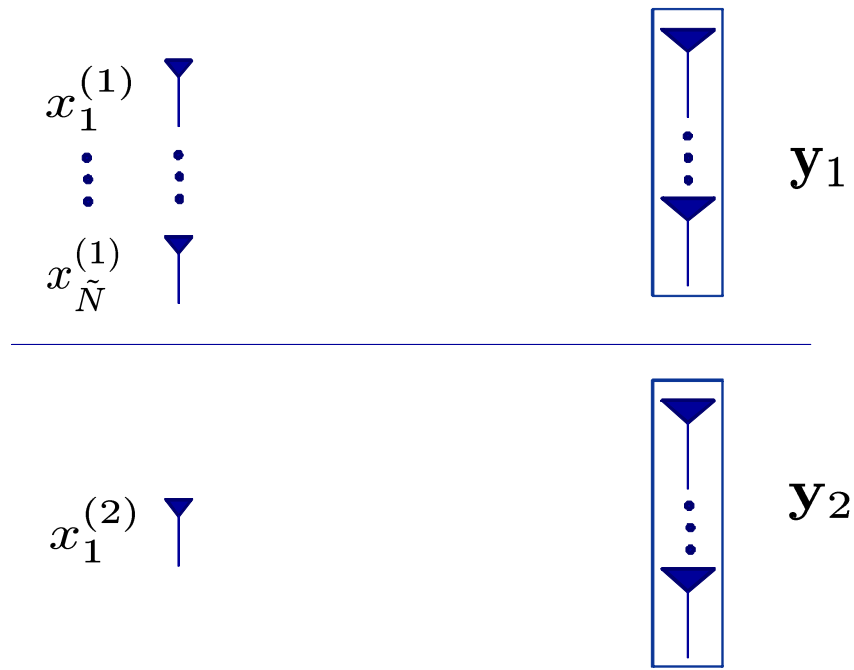} \caption{The two-cell IMAC model defined in
Section~\ref{sect:Upperbound_DoF}.} \label{FIG:upper_fig}
\end{center}
\end{figure}

\end{document}